\documentclass[pra,twocolumn,superscriptaddress]{revtex4-1}
\pdfoutput=1

\usepackage[verbose=true]{microtype}
\usepackage{amsmath}
\usepackage{amssymb}
\usepackage{txfonts}
\usepackage{mathtools}
\usepackage{color}
\usepackage[colorlinks,linkcolor=black,citecolor=black,linktocpage,hypertexnames=true,pdfpagelabels]{hyperref}
\usepackage[amsmath,thmmarks,hyperref]{ntheorem}
\usepackage[capitalise]{cleveref}
\usepackage[all]{xy}
\usepackage{subcaption}
\usepackage{array}
\usepackage{enumerate}
\usepackage{enumitem}
\usepackage{tikz}

\newtheorem{theorem}{Theorem}

\newtheorem{corollary}[theorem]{Corollary}
\newtheorem{proposition}[theorem]{Proposition}
\newtheorem{definition}[theorem]{Definition}

\usepackage{ntheorem}

\newtheorem{theorem-nn}[theorem]{Theorem}

\newtheorem*{theorem:CL}{\cref{thm:CL}}

\theorembodyfont{\normalfont}
\theoremsymbol{$\Box$}

\theorembodyfont{\normalfont}
\theoremsymbol{$\Box$}
\newtheorem{example}[theorem]{Example}

\theorembodyfont{\normalfont}
\theoremsymbol{$\Box$}
\newtheorem{remark}[theorem]{Remark}

\theoremstyle{nonumberplain}
\theorembodyfont{\normalfont}
\theoremsymbol{$\Box$}
\newtheorem{proof}{Proof}

\theoremstyle{nonumberbreak}
\theorembodyfont{\normalfont}
\theoremsymbol{$\Box$}

\renewcommand{\epsilon}{\varepsilon}
\newcommand{\E}{\mathcal{E}}

\newcommand{\be}{\begin{equation}}
\newcommand{\ee}{\end{equation}}
\newcommand{\bea}{\begin{eqnarray}}
\newcommand{\eea}{\end{eqnarray}}
\newcommand{\ben}{\begin{enumerate}}
\newcommand{\een}{\end{enumerate}}

\newcommand{\ra}{\rangle}
\newcommand{\la}{\langle}
\newcommand{\mc}{\mathcal}

\newcommand{\tr}{\mathrm{Tr}}
\newcommand{\diag}{\mathrm{diag}}

\newcommand{\nn}{\nonumber}
\newcommand{\lb}{\label}

\DeclareMathOperator*{\Moplus}{\text{\raisebox{0.25ex}{\scalebox{0.7}{$\bigoplus$}}}}

\begin{document}

\title{Continuum limits of Matrix Product States}

\author{Gemma De las Cuevas}
\address{Institut f\"ur theoretische Physik, Universit\"at Innsbruck, 
Technikerstr.\ 21a, 6020 Innsbruck, Austria} 
\author{Norbert Schuch}
\address{Max Planck Institute for Quantum Optics, 
Hans-Kopfermann-Str.\ 1, 85748 Garching, Germany}
\author{David Perez-Garcia}
\address{Departamento de An\'alisis Matem\'atico and IMI, 
Universidad Complutense de Madrid, 28040 Madrid, Spain}
\address{ICMAT, C/ Nicol\'as Cabrera, Campus de Cantoblanco, 28049 Madrid,
Spain}
\author{J. Ignacio Cirac}
\address{Max Planck Institute for Quantum Optics, 
Hans-Kopfermann-Str.\ 1, 85748 Garching, Germany}

\begin{abstract}
We determine which translationally invariant matrix product states have a
continuum limit, that is, which can be considered as discretized versions
of states defined in the continuum.  To do this, we analyse a
fine-graining renormalization procedure in real space, characterise the
set of limiting states of its flow, and find that it strictly contains the
set of continuous matrix product states. We also analyse which states have
a continuum limit after a finite number of a coarse-graining
renormalization steps. We give several examples of states with and without
the different kinds of continuum limits.
\end{abstract}

\maketitle

\section{Introduction}

The quest for continuum limits of discrete theories is a central topic in high energy physics \cite{We95,Pe95b}  and condensed matter physics \cite{Sa11,Fr13b}.  In many cases, the continuum limit of a theory is obtained after a renormalization process, where the lattice constant (which provides an energy cutoff) is taken to zero. This occurs, for instance, in quantum lattice models, where the continuum limit is the desired quantum field theory and the renormalization involves the redefinition of the parameters of the Hamiltonian describing the model. The question of whether a particular quantum lattice model possesses the correct continuum limit under renormalization is of central interest in several fields of quantum physics.

Tensor networks have proven to be useful tools to study strongly correlated systems in quantum lattice models \cite{Or14b,Sc05b,Ve08}. In fact, in one spatial dimension, matrix product states (MPS) \cite{Fa92,Pe07}, a special kind of tensor network states (TNS), provide the most powerful technique to study such systems. In contrast to some traditional approaches to describe quantum many-body systems where the {\em Hamiltonian} (or the action) is the central object of study, the theory of tensor networks concentrates on the description of quantum many-body {\em states}. The reason is that they are completely characterised (for homogeneous systems) by a simple tensor, whose rank depends on the coordination number of the lattice. 
The fact that ground states (vacuum) and low energy excitations of local theories are expected to have very little entanglement 
makes tensor networks efficient tools for describing them. Furthermore, they can be used as toy models to analyse complex phenomena associated to topology \cite{Sc10b}, symmetry protection \cite{Ch11,Sc11}, or even chirality \cite{Wa14}, in relatively simple terms.

Renormalization procedures in tensor networks and, in particular, in MPS,
have played an important role in the development of various methods
associated to them. The renormalization of a TNS provides a coarse-grained
description of the state and, in the case of MPS, flows to a very specific
family of states that can be fully characterised \cite{Ve05}. In fact,
these fixed points of the renormalization procedure have been used to
obtain a classification of the (gapped) quantum phases of spin chains in
one spatial dimension \cite{Ch11,Sc11}.

In this work, we investigate how the same renormalization procedure can give a rigorous method to obtain the continuum limit of an MPS. That is, we consider the inverse procedure of coarse-graining, i.e.\ fine-graining, and investigate to what extent it converges, and to which kind of states. Or, more boldly stated, we solve the following problem: given an MPS, when is it the coarse-grained picture of the vacuum of a quantum field theory in one spatial dimension? We will then say that such an MPS has a continuum limit (CL).

To be specific, we consider a fine-graining procedure  such that the state is translationally invariant at all steps. Moreover, each fine-graining step is carried out by some isometry, which can differ from step to step. 
As a consequence, the finer state is in fact the same state as the original one, but  written in a finer basis, i.e.\ a basis with more sites. 
Thus, our definition of CL is very restrictive and can be seen as a first step toward  the study of CLs in more general settings.

Now, while it is clear that some states must have a CL in the sense specified below, it is also clear some others will not. For instance, a ferromagnetic state $|0,\ldots,0\rangle$ clearly has a CL, which is the vacuum of a non-interacting theory in the continuum. In contrast, a superposition of two antiferromagnetic states,
 \be
 \label{eq:AF}
 |\Psi_{\rm af}\ra = \frac{1}{\sqrt{2}} \left( |0,1,0,1,\ldots\ra +
 |1,0,1,0,\ldots\ra \right),
 \ee
will not have such a limit, since there exists no (translationally
invariant) state such that if we coarse-grain it, we obtain $|\Psi_{\rm
af}\ra$. 
But, what about states like the AKLT \cite{Af88}, the cluster
state \cite{Br01}, or other prominent states found in the field of
condensed matter or quantum information theory? 

On the other hand, by flipping every second spin in the $z$ direction,  $|\Psi_{\rm af}\ra$ is mapped to a superposition of the two ferromagnetic states, $|0,0,\ldots,0\ra +|1,1,\ldots,1\ra  $, which  has a CL. 
While in our  definition of CL we only allow to apply  operations (isometries) which are the same on every site,  this restriction is lifted in our second definition of CL, called the coarse continuum limit. In the latter, we first coarse grain the state, and then take the CL of the coarse-grained state. 
Thus, $|\Psi_{\rm af}\ra$  has a coarse CL, but does every state have a coarse CL?

In this paper we give an answer to these questions by determining the conditions for a state to have a CL. We also characterize which are the set of states of the quantum field theory which are the CL of an MPS. We find that such a set contains continuous MPS (cMPS) \cite{Ve10,Ha13}, as one would expect, but it also contains some extensions that have not been encountered so far in the study of TNS. We finally show that there exist states that do not possess a CL even if we first coarse-grain any finite number of times, i.e.\ not every state has a coarse CL. We note that different continuum limits of quantum lattice systems have been considered in Ref.\ \cite{Os-cl}, and tensor network descriptions of quantum field theories have been studied in Refs.\ \cite{Je15}.

This paper is organised as follows. 
In \cref{sec:cl} we define and characterise the CL of MPS. 
In \cref{sec:ccl} we define and characterise the coarse CL of MPS, present examples of states with either kind of CL, and compare the two CLs.   
In \cref{sec:concl} we conclude. We leave the proof of the main result (\cref{thm:CL}) to \cref{app:proof-dyadic}.

\section{Continuum limit}
\label{sec:cl}

In this section we present our work on the CL of an MPS. 
We will first explain the setting of our problem (\cref{ssec:setting}), 
define and characterise $p$-refining (\cref{ssec:prefining}), and finally define and characterise the CL of an MPS (\cref{ssec:cl}).

\subsection{The setting}
\label{ssec:setting}

Our starting point is a three-rank tensor $A = \{A^i\in \mathcal{M}_D\}_{i=1}^d$, where $\mc{M}_D$ denotes the set of $D\times D$ complex matrices, $D$ is called the bond dimension, and $d$ the physical dimension, both of which are assumed to be fixed and finite.
$A$ generates a translationally-invariant (TI) MPS
\be
|V_N(A)\ra :=  \sum_{i_1,\ldots,i_N}
\tr(A^{i_1}A^{i_2}\cdots A^{i_N})|i_1,\ldots, i_N\ra
\label{eq:V}
\ee
for every $N\in \mathbb{N}$, as well as the family
\be
\mathcal{V}(A) := \left\{ |V_N(A)\ra \right\}_{N\in \mathbb{N}} .
\ee
As the tensor $A$ completely determines all the properties of the MPS it generates, when developing the theory of MPS one works directly with such a tensor.

The  \emph{transfer matrix} of $\mc{V}(A)$, $E_A$, is defined as \cite{Ve05}
 \be
 E_A = \sum_{i=1}^d A^i \otimes \bar A^{i},
 \label{eq:tm}
 \ee
where the bar indicates complex conjugation. 
Note that $E_A$ is (a matrix representation of) the completely positive map (CPM) $\E(\cdot) = \sum_{i=1}^d A^i \cdot A^{i\dagger}$, and it is independent of any isometry applied to the physical index $i$. 
In Ref.\ \cite{De17} we showed that, without loss of generality, 
$A$ can be taken to be in irreducible form, that is, 
$A^i = \Moplus_j \mu_j A_j^i$, where $\mu_j>0$, and each $E_{A_j}$ is an irreducible CPM (i.e.\ a CPM with a non-degenerate eigenvalue 1, but which can have other eigenvalues of modulus 1).
Moreover, $E_A$ can be taken to be a quantum channel (i.e.\ a trace-preserving (TP) CPM). 
We will thus indistinctively call $E_A$ a transfer matrix or a quantum channel. 
If clear from the context, we will simply denote it by $E$. 
\subsection{Definition and characterisation of $p$-refining}
\label{ssec:prefining}

The renormalization procedure introduced in Ref.\ \cite{Ve05} basically maps $|V_N(A)\ra$ to 
\be
|V_{N}(B)\ra = (W^{ \dagger})^{\otimes N}|V_{pN}(A)\ra \quad \forall N 
\label{eq:coarsegraining}
\ee
where $p>1$ is an integer and $W:\mathbb{C}^d\to (\mathbb{C}^d)^{\otimes p}$
is an isometry. 
We now introduce the inverse step.
\begin{definition}
We say that $\mc{V}(B)$ can be \emph{$p$-refined} if there exists another tensor $A$ and an isometry  $W$ such that
\be
|V_{pN}(A)\ra = W^{\otimes N}|V_N(B)\ra \quad \forall N .
\label{eq:fine}
\ee
\end{definition}

Clearly, if $\mc{V}(B)$ can be $p$-refined with the isometry $W$, then it can also be $p$-refined with the isometry $U^{\otimes p}W$, where $U$ is a unitary. 
We thus call two isometries $W,W'$ inequivalent if there is no unitary $U$ such that $W' = U^{\otimes p}W$. 
Similarly, we say that $\mc{V}(B)$ can be $p$-refined in $r$ inequivalent ways if it can be $p$-refined with $r$ inequivalent isometries.

In Ref.\ \cite{De17} we showed that $\mathcal{ V}(B)$ can be $p$-refined if and only if $E_B$ is \emph{$p$-divisible}; 
that is, if there exists a quantum channel $E_{p}$ such that $E_p^p=E_B$. 
Moreover, the number of inequivalent ways of $p$-refining a state is precisely given by the number of $p$th roots of its transfer matrix which are also a transfer matrix. 
The divisibility of quantum channels has been analyzed in Refs.\ \cite{Ho87,De89b,Wo08} in the context of Markovian evolution of quantum systems. 
In particular, there exist channels that are are not $p$-divisible for any $p$ \cite{Wo08}. 
This automatically implies that there are states that cannot be refined at all \cite{Wo08b} -- 
we will see two examples thereof in \cref{ex:holevo} and \cref{ex:aklt}.
In \cref{rem:roots} we will mention examples of states that can be refined in several inequivalent ways. 

\subsection{Definition and characterisation of continuum limit }
\label{ssec:cl}

One could define the CL of an MPS as the limiting point
of the $p$-refining procedure. However, such definition would not be satisfactory
since there are states that can be refined but that should not have  a CL. 
This can be illustrated by means of the antiferromagnetic state of  \cref{eq:AF}, 
which can be $3$-refined infinitely many times with the isometry $W=|0,1,0\rangle\langle 0| + |1,0,1\rangle\langle 1|$. 
 However, it is clear that it cannot exist in the continuum. (This state will be more thoroughly analysed in \cref{ex:aferro}).

To deal with this problem, we notice that if we had a CL, it would be 
reasonable to demand that the limit should not depend on whether we block
a few spins when we are close to that limit. Differently speaking,
introducing an intermediate coarse-graining step should not affect the
form of the CL.  This e.g.\ rules out the antiferromagnetic state: 
In \cref{eq:AF}, if we $3$-refine many
times 
with the isometry
$W= |0,1,0\ra\la 0| + |1,0,1\ra\la 1| $
and then block 2 spins, 
with the isometry
$W' = |0,1\ra\la 0| + |1,0\ra\la 1| $
we obtain a GHZ-like state \cite{Gr89b},
$|0,0,\ldots,0\ra + |1,1,\ldots,1\ra $, which is very different from the
fixed point if we had not blocked. This motivates the following definition.

\begin{definition}
\lb{def:CL}
We say that \emph{$\mathcal{V}(B)$ has a continuum limit (CL)} if there is a $p>1$
such that the procedure of $p$-refining $\ell$ times followed by the blocking of $n_{\ell}\in \mathbb{N}$ of the resulting spins 
converges in $\ell$, as long as $(n_\ell/p^{\ell})_\ell\to 0$ as $\ell\to\infty$.
\end{definition}

Note that $(n_\ell/p^{\ell})_\ell$ denotes the infinite sequence whose elements are $n_\ell/p^{\ell}$ with $\ell\in \mathbb{N}$.
We now want to characterise which states have a CL in terms of the divisibility properties of its transfer matrix.
The requirement that the state be $p$-refinable infinitely many times translates to the requirement that its transfer matrix $E$ be \emph{$p$-infinitely divisible}. 
This means that $E$ is $p^\ell$-divisible for any $\ell \in \mathbb{N}$, that is, that for any $\ell \in \mathbb{N}$ there is a quantum channel $E_{p^\ell}$ such that $E_{p^\ell}^{p^\ell} =E$.
Note that a quantum channel $E$ is called \emph{infinitely divisible} if it is $n$-divisible for any $n$, i.e.\ $E = E_n^n$ for all $n\in  \mathbb{N}$ \cite{Wo08}. 

We also need to characterise the condition of stability of the limiting procedure  under blocking (cf.\  \cref{def:CL}). To this end, we  introduce the following function (see, e.g., Ref.\ \cite{Yu76}). 
Let $E$ be a $p$-infinitely divisible quantum channel  and let $\{E_{p^\ell}\}_{\ell \in \mathbb{N}}$ be a set of roots which are quantum channels themselves. We define the function  $f_{p,E}$ as 
 \be
 f_{p,E}(n,\ell)=E_{p^\ell}^n,
 \label{eq:fpE}
 \ee
where $n,\ell\in \mathbb{N}$. 
Now, we say that $f_{p,E}$ is \emph{continuous at 0} if there exists a set $\{E_{p^\ell}\}_{\ell\in \mathbb{N}}$ and  a matrix  $Q$,
such that for all sequences $\{n_k,\ell_k\}_{k=1}^\infty$ fulfilling
$\lim_{k\to\infty}n_k/p^{\ell_k}= 0$,
it holds that $\lim_{k\to\infty}f_{p,E}(n_k,\ell_k) = Q$.
Thus, the existence of a CL is equivalent to the existence of  a $p>1$ such that $E_B$ is $p$-infinitely divisible, and an $f_{p,E_B}$ which is continuous at zero. With this, we can characterise the set of MPS with a CL.

\begin{theorem}[Main result]
\label{thm:CL}
Given $\mathcal{V}(B)$ with $B$ in irreducible form, the following statements are equivalent:
\begin{enumerate}
\item $\mathcal{V}(B)$ has a CL.
\item $E_B$ is infinitely divisible. 
\label{infdiv}
\item There is a quantum channel $P$ and a Liouvillian of Lindblad form $L$ such that $E_B=Pe^L$,  $P^2=P$ and $PLP=PL$. \label{hd}
\end{enumerate}
\end{theorem}
The proof is given in Appendix~\ref{app:proof-dyadic}.

Note that the last item fully
characterizes all possible CLs.  If $P=\openone$, the corresponding
transfer matrix $e^L$ coincides with that of a TI cMPS. Thus, as expected,
all TI cMPS can be limits of TI MPS. However, for $P\ne \openone$, other
states than cMPS appear as possible CLs.
Note also that one can easily see from condition \ref{hd} of \cref{thm:CL} that the limit is smooth, 
as $\lim_{t\to 0} E^t =\lim_{t\to 0} Pe^{tL} =P$.
Finally, note that from \cref{thm:CL} and the results of \cite{De17} it follows that if 
$\mc{V}(B)$ has a CL, then $\mc{V}(B)$ can be $p$-refined for any $p>1$.

\section{Coarse continuum limit}
\label{sec:ccl}

We now present a more relaxed definition of a CL of an MPS, which we call the coarse CL. 
We will first define and characterise it (\cref{ssec:ccl}),
give several examples of states with or without a coarse CL (\cref{ssec:ex}),
and finally use these examples to compare the two notions of CL (\cref{ssec:comparison}).

\subsection{Definition and characterisation}
\label{ssec:ccl}

We have seen that to obtain a meaningful definition of a CL we have to impose that we can block towards the end of the refinement, and still obtain the same limit. 
We can thus ask what happens if we allow for blocking before the refinement.
For example, by blocking 2 sites of the antiferromagnetic state (\cref{eq:AF}), we obtain the ferromagnetic state, which has a trivial CL.
This motivates the following definition (see \cref{fig:RG-general}).

\begin{figure}[t]\centering
\includegraphics[width=0.95\columnwidth]{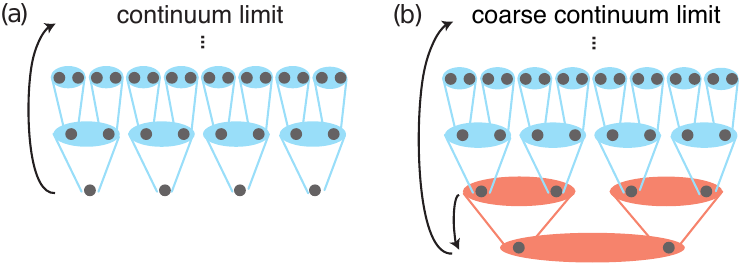}
\caption{
Sketch of (a) the continuum limit and (b) the coarse continuum limit.  }
\label{fig:RG-general}
\end{figure}

\begin{definition}
\label{def:dcl}
We say that $\mc{V}(A)$ has a \emph{coarse CL}
if there is a $\mc{V}(B)$ and an $n\in \mathbb{N}$ such that
$\mc{V}(A)$ is the $n$-refinement of $\mc{V}(B)$, and  $\mc{V}(B)$ has a CL.
\end{definition}

Note that every state $\mc{V}(A)$ is the $p$-refinement of some other state $\mc{V}(B)$, i.e.\ given $A$ and $p$ there is always  an isometry $W$ and a tensor $B$ that satisfies Eq.\ \eqref{eq:coarsegraining}. 
Moreover, the process of ``coarse-graining'' $p$ sites (the opposite of $p$-refining) is essentially unique; more precisely, different isometries  will give rise to tensor $B$'s which are related by a unitary matrix in the physical index, as shown in Ref.\ \cite{Ve05}. 
This is again best understood at the level of the transfer matrix: coarse-graining $p$ sites corresponds to 
taking the $p$th power of the transfer matrix, which gives a unique result, and which always corresponds to a valid transfer matrix. This is to be contrasted with $p$-refining, which is only possible if there is at least one $p$th root of $E$ which is a valid transfer matrix, and in case there is, they may be multiple such roots. 

The following characterisation is immediate from the above  results.

\begin{corollary}
\label{thm:dcl}
$\mc{V}(A)$ has a coarse CL if and only if there exists an $n\in \mathbb{N}$ such that
$E_A^n$ is infinitely divisible.
\end{corollary}

\begin{remark}[Computational complexity]
What is the computational complexity of deciding whether a state has a (coarse) CL? 
Concerning the CL, deciding infinite divisibility is at least as hard as deciding Markovianity, 
since the latter amounts to deciding the former together with being full rank (see condition 3 of \cref{thm:CL}), 
and being full rank can be decided efficiently. 
Deciding Markovianity has been formulated as an integer Semidefinite Program for fixed input dimension \cite{Wo08b}, 
and shown to be NP-hard as a function of the bond dimension \cite{Cu12b}.
Concerning the coarse CL, to the best of our knowledge, 
the computational complexity of determining whether, given a channel $E$, there is some $n\in \mathbb{N}$ such  $E^n$ is infinitely divisible is not known. 
\end{remark}

\subsection{Examples}
\label{ssec:ex}

We now  present several examples of states with either kind of CL which illustrate \cref{thm:CL} and \cref{thm:dcl}.

\begin{example}[The ferromagnet]
\label{ex:ferro}
Let us start with an equal superposition of $m$ ferromagnetic states, 
\be
|V_N(B)\ra = \sum_{i=0}^{m-1} |i,i\ldots i \ra ,
\ee
which is given by the tensor $B = \{B^i\in \mc{M}_D\}_{i=0}^{m-1}$, where $B^i = |i\ra\la i|$ for $i=0,1,\ldots, m-1$.  
$\mc{V}(B)$ can be $p$-refined into $p$ copies of itself for any $p$ with 
$W =\sum_{i=0}^{m-1} |i,i,\ldots \ra \la i |$, 
and this is also true after the blocking of an arbitrary number of spins. 
Equivalently (see \cref{thm:CL}), the transfer matrix  
\be
E_{\mathrm{f}}= \sum_{i=0}^{m-1} |i,i\ra\la i,i|  
\label{eq:Ef}
\ee
 is a projector, thus it is infinitely divisible, and thus the state has a CL.
Recall that the transfer matrix (cf.\ \eqref{eq:tm}) acts on the auxiliary space, whereas $|V_N(B)\ra $ is a state living in the physical space. 
\end{example}

\begin{example}[The antiferromagnet]
\label{ex:aferro}
Consider an equal superposition of $m$ antiferromagnetic states,
\be
|V_m(B)\ra = \sum_{i=0}^{m-1} |i, i+1, \ldots, i+m-1\ra, 
\ee
where the sum is modulo $m$,
(and similarly for $N$ multiple of $m$, and $|V_N(B)\ra=0$ otherwise),
which is given by 
$B^i = |i\ra\la i+1|$ for $i=0,1,\ldots, m-1$.
$\mc{V}(B)$ can be $p-$refined into $p$ copies of itself, with $p=m+1$, with the isometry
\be
W = \sum_{i=0}^{m-1} |i,i+1, \ldots, i+m-1, i\ra\la i|.
\label{eq:Waf}
\ee
However, as we have discussed, this state does not have a CL, since the limit of this refinement is not stable under blocking. 
Equivalently (see \cref{thm:CL}), the transfer matrix $E_{\textrm{af}}$ is $p$-infinitely divisible with $p=m+1$, since 
\be
E_{\textrm{af}} = \sum_{i=0}^{m-1} |i,i\ra\la i+1,i+1| =E_{\textrm{af}}^{m+1}, 
\label{eq:Eaf}
\ee
but it is not infinitely divisible, since it does not have, e.g., an $m$th
root which is a quantum channel. To see the latter, note that the non-zero
part of the spectrum of $E_{\textrm{af}}$ is $\{e^{2\pi i
r/m}\}_{r=0}^{m-1}$, and thus for its $m$th root $\{e^{2\pi i
\ell_r/m^2}\}_{r=0}^{m-1}$ (with e.g.\ $\ell_1$ coprime to $m^2$),
whereas the set of eigenvalues of modulus 1 of a quantum channel needs to
be of the form $\{e^{2\pi i r/n}\}_{r=0}^{n-1}$ for some $n$ \cite{Wo11}.
On the other hand, $\mc{V}(B)$ has a coarse CL,
since after blocking $m$ sites we obtain the ferromagnet of \cref{ex:ferro}.
\end{example}

\begin{example}[A deformed antiferromagnet]
We consider the tensor $B(\alpha)$ (with $0<\alpha<1$) 
\begin{eqnarray}
&&B^0(\alpha) =  \sqrt{\alpha} \: |0\ra\la 1| + \sqrt{1-\alpha}\: |1\ra\la 0| ,\\
&&B^1(\alpha) =B^0(\alpha)^t,  \quad
\end{eqnarray}
where $t$ denotes transpose. The corresponding state has periodicity 2, as for even $N$ we have that 
\be
|V_N(B(\alpha))\ra = |\mu_0,\mu_1,\mu_0,\mu_1\ldots\ra +
|\mu_1,\mu_0,\mu_1,\mu_0\ldots\ra , 
\ee
where $|\mu_i\ra$ is  shorthand for $|\mu_i(\alpha)\ra$, and 
\be 
|\mu_i(\alpha)\ra=\sqrt{\alpha} |i\ra + \sqrt{1-\alpha} |i+1\ra
\label{eq:mu}
\ee
for $i=0,1$, where the sum on $i$ is mod 2.
Now, let 
\be
g_{\pm}(\alpha) = \frac{1}{2}\left(1\pm\sqrt{1-(4\alpha(1-\alpha))^{1/3}}\right) .
\ee
Then $\mc{V}(B(\alpha))$ can be $3$-refined into $\mc{V}(B(g_+(\alpha)))$ or $\mc{V}(B(g_-(\alpha)))$. 
The corresponding isometries are given by 
\begin{eqnarray}
W_{\pm}= \frac{1}{1-\lambda(\alpha)^2} (&&
|\nu_0^{\pm}  \ra \la \mu_0| + 
|\nu_1^{\pm}  \ra  \la \mu_1| \nonumber\\
&&-\lambda(\alpha) |\nu_0^{\pm} \ra \la \mu_1| 
-\lambda(\alpha) |\nu_1^{\pm} \ra \la \mu_0| 
), 
\end{eqnarray}
where 
$|\nu_i^{\pm}  \ra = |\mu_i(g_\pm(\alpha)),\mu_{i+1}(g_\pm(\alpha)),\mu_i(g_\pm(\alpha))\ra$
for $i=0,1$ 
where the sum on $i$ is modulo 2, 
and 
\be
\lambda(\alpha)=2\sqrt{\alpha(1-\alpha)} .
\label{eq:lambda}
\ee
However, this refinement is not stable under the blocking of two spins, since that would give rise to a state without periodicity. 
Equivalently (see \cref{thm:CL}), the transfer matrix $E_{B(\alpha)}$ is 3-infinitely divisible but not infinitely divisible. 
To see this, note that  in the Pauli basis (which is defined as usual, namely 
$\openone = |0\ra\la 0| +|1\ra\la 1| $, 
$X = |0\ra\la 1| +|1\ra\la 0|$,
$Y= -i|0\ra\la 1| +i|1\ra\la 0|$,
$Z = |0\ra\la 0| -|1\ra\la1|$) we have that 
\be
E_{B(\alpha)} = \diag(1,\lambda(\alpha), -\lambda(\alpha),-1). 
\label{eq:Espectrum}
\ee
Therefore, $E_{B(\alpha)} = E_{B(g_{\pm}^{\ell}(\alpha))}^{{3^{\ell}}}$ for all natural $\ell$, where 
$ E_{B(g_{\pm}(\alpha))} =\diag(1,\lambda(g_{\pm}(\alpha)), -\lambda(g_{\pm}(\alpha)),-1)$
where we choose either $g_+$ or $g_-$ for both eigenvalues, 
and $g_{\pm}^{\ell}$ denotes the $\ell$-fold application of the map $g_{\pm}$.
Yet, $E_{B(\alpha)}$ does not have, e.g., a square root which is a quantum channel, since the spectrum of a channel needs to be closed under complex conjugation, which is impossible given \eqref{eq:Espectrum}.
Thus, this state does not have a CL. 
However, after blocking two sites we obtain a Markovian transfer matrix, 
namely $E_{B(\alpha)}^2 = e^{L}$ with
$\mc{L}(\rho) = -\ln(\lambda(\alpha)) (Z\rho Z-\rho)$.
Thus, this state has a coarse CL.
\end{example}

 \begin{example}[The cluster state]
Consider the one-dimensional (1D) cluster state $\mathcal{V}(A)$ \cite{Br01}, which is obtained with the tensor
\be
A^1 = |1\ra \la +|, \quad
A^2 =|0\ra\la -|,
\ee
where
$|\pm\ra = (|0\ra \pm |1\ra)/\sqrt{2}$ \cite{Pe07}.
The transfer matrix 
\be
E_A = |0,0\ra \la -,-| +|1,1\ra \la +,+|
\ee
has eigenvalues $(1,0,0,0)$, but the eigenvalue 0 is associated to a non-trivial Jordan block. This block does not have a $p$th root for any $p$ (see Definition 1.2.\ of Ref.\ \cite{Hi08}), and thus $\mc{V}(A)$ cannot be $p$-refined  for any $p$.
However, $ E^{2} =(1/2) (|0,0\ra + |1,1\ra)(\la 0,0|  + \la 1,1|) $ is a projector, and hence has a trivial CL. 
Thus, the 1D cluster state has a coarse CL. 
\end{example}
 
 \begin{example}[The Holevo--Werner channel]
 \label{ex:holevo}
Consider the Holevo--Werner channel for qubits,
$\E(\rho)=\frac{1}{3} \left(\rho^t+\tr(\rho) \openone \right)$,
where   $\rho^t$ denotes its transpose.
The corresponding state is given by the tensor
\be
A^1 = \sqrt{\frac{2}{3}} |0\ra\la 0|,
\quad
A^2 = \sqrt{\frac{2}{3}} |1\ra \la 1|,
\quad
A^3 = \frac{1}{\sqrt{3}} X.
\ee
In the Pauli basis, $E = \diag(1,1/3,-1/3,1/3)$.
This channel cannot be expressed as a non-trivial composition of two quantum channels (even if these two are different) \cite{Wo08}, and thus $\mc{V}(A)$ cannot be $p$-refined for any $p$.
However, 
$E^2$ is Markovian, namely $\E^2 = e^{\mc{L}_\gamma}$, with
 \be
\quad \mc{L}_\gamma (\rho)= \gamma \left(
 X\rho X +  Y\rho Y + Z\rho Z-3\rho\right),  \:\: \gamma=\ln (9)/4.  \label{eq:Mark}
 \ee
Thus this state has a coarse CL.
More generally, note that every odd power of $E$ is not infinitely divisible,
$\det(E^n)<0$ for  odd $n$ (see Proposition 15 of \cite{Wo08}), 
and every even power of $E$ is Markovian.
\end{example}

\begin{example}[AKLT state]
\label{ex:aklt}
Consider the AKLT state \cite{Af88}, which is described in terms of the tensor
\be
A^1 = \frac{1}{\sqrt{3}} Z, \quad
A^2 = \sqrt{\frac{2}{3}} |1\ra \la 0|, \quad
A^3 = - \sqrt{\frac{2}{3}} |0\ra\la 1|.
\ee
In the Pauli basis, $E = \diag\left(1,-1/3,-1/3,-1/3\right).$
We thus have that $\det (E) =-1/27$, and the channel cannot be expressed as a non-trivial composition of two quantum channels  \cite{Wo08}. Thus the AKLT state cannot be $p$-refined for any $p$.
However, $\E^2 = e^{\mc{L}_\gamma}$, with $\mc{L}_\gamma$ given by \eqref{eq:Mark}. 
More specifically, 
$E^2 = \sum_{i=1}^4 B^i\otimes \bar B^i$, 
with 
\begin{align}
&B^1 = \sqrt{q} I, \quad
B^2 = \sqrt{\frac{1-q}{3}} X, \nonumber\\
&B^3 = \sqrt{\frac{1-q}{3}}Y, \quad
B^4 = \sqrt{\frac{1-q}{3}}Z, 
\end{align}
with $q=1/3$. This state can be $p$-refined for any $p$ into a state with the same matrices, but with $q$ replaced  by $q_p = (1+3^{(p-2)/p})/4$.
Thus the AKLT state has a coarse CL.
\end{example}

\begin{remark}[Multiple roots of the transfer matrix]
\label{rem:roots}
\cref{ex:ferro} and \cref{ex:aferro} illustrate that the transfer matrix of the ferromagnet with $m$ states (Eq.\ \eqref{eq:Ef}) has two $p=m+1$ roots which correspond to a transfer matrix, namely itself, and the transfer matrix of the antiferromagnet (Eq.\ \eqref{eq:Eaf}). These correspond to the two inequivalent ways of $p$-refining the state.

Similarly, \cref{ex:holevo} and \cref{ex:aklt}  illustrate that the depolarizing channel
$E= e^{L_\gamma}$ with $L_\gamma$ given in \eqref{eq:Mark}, 
has three square roots which are valid quantum channels:
the Markovian one ($e^{L_{\gamma}/2}$),
the Holevo--Werner channel,
and the transfer matrix corresponding to the AKLT state.
Only the Markovian root can be further refined, and 
thus the state corresponding to $E= e^{L_\gamma}$  has a CL. 
\end{remark}

Finally, we give an example of a state without a coarse CL. 

\begin{example}[A state without a coarse CL]
\label{ex:pancake}
Consider the family of qubit channels of the form $E =1\oplus \Delta$ in the Pauli basis,
with $\Delta$ positive definite and with eigenvalues $\lambda_1\geq \lambda_2\geq \lambda_3$.
We claim that if $0<\lambda_3<\lambda_1\lambda_2$, then $E^n$ is not infinitely divisible for any finite $n$.
To see this, note that by Theorem 24 in Ref.\ \cite{Wo08} $E$ is not infinitesimal divisible, and this is preserved under powers. Since infinitely divisible channels are a subset of infinitesimal divisible channels \cite{Wo08}, it follows that  the state corresponding to this transfer matrix does not have a coarse CL.

Take for example $\Delta$ diagonal and $\lambda_1=\lambda_2= a$, $\lambda_3 = a^2/2$ (with $0<a\leq 2-\sqrt{2}$, see the proof of \cref{prop:finite}). 
The corresponding tensor is given by
\bea
&&A^1= \sqrt{\tfrac{2+4a+a^2}{8}} \openone \qquad
A^2= -\sqrt{\tfrac{2-4a+a^2}{8}} Z\qquad
\nn\\
&&
A^3= \sqrt{\tfrac{2-a^2}{8}}|1\ra\la0| \qquad
A^4= \sqrt{\tfrac{2-a^2}{8}}|0\ra\la1| . \label{eq:pancake}
\eea
Note that $\lim_{n\to \infty} E^n = C$, where $C$ is the completely depolarizing channel, $\mathcal{C}(\rho)=\tr(\rho)\openone/2$. The latter is in the closure of the set of Markovian channels (e.g.\ $ \mathcal{C} = \lim_{\gamma\to \infty} e^{\mc{L}_\gamma}$, with $\mc{L}_\gamma$ given in \eqref{eq:Mark};
see Fig.~\ref{fig:geometry-markovian}).
\end{example}

\begin{figure}[t]
\begin{center}
\includegraphics[width=0.3\textwidth]{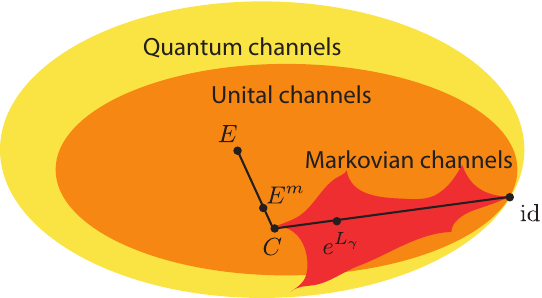}
\end{center}
\caption{Sketch of part of the geometry of qubit channels.
The volume of the sets is drawn arbitrarily.}
\label{fig:geometry-markovian}
\end{figure}

\subsection{Comparison between the two continuum limits}
\label{ssec:comparison}

The previous examples allow us to compare the two CLs.
Let $\mc{C}_D$ and $\mc{C}_D^{\textrm{coarse}}$ denote the set of families of states  $\mc{V}(A)$ of bond dimension $D$ with a CL and a coarse CL, respectively. 

\begin{proposition}
\label{prop:sets}
For every bond dimension $D$,
\begin{enumerate}[label=\arabic{enumi}.,ref=\arabic{enumi}.]
\item \label{im:1sets}
$\mc{C}_D$  is strictly included in $\mc{C}_D^{\textrm{coarse}}$, and
\item \label{im:2sets}
There are states  not in  $\mc{C}_D^{\textrm{coarse}}$.
\end{enumerate}
\end{proposition}

\begin{proof}
\ref{im:1sets} That $\mc{C}_D$ is included in $\mc{C}_D^{\textrm{coarse}}$
is trivial from the definition, and
 for $D=2$, that the inclusion is strict  follows e.g.\ from \cref{ex:holevo}. 
For $D=2$, the second claim is proven by \cref{ex:pancake}.
In both cases, the extension to higher $D$ follows trivially by embedding $\mathcal{M}_2$ into  $\mathcal{M}_D$, for example, as  $\mathcal{M}_D = \mathcal{M}_2 \oplus 0_{D-2}$, where $0_{D-2}$ is the 0 matrix.
\end{proof}

We also gain the following insight from \cref{ex:pancake}.

\begin{proposition}\label{prop:finite}
There are states that can be $p$-refined only a finite number of times.
\end{proposition}

\begin{proof}
Consider the family of channels whose Lorentz normal form \cite{Wo08} is given by $E(a,\eta):=\diag(1,a,a,\eta a^{2})$, with $a\in (0,1]$ and $\eta\in (0,1)$.
It is easy to see that $E(a,\eta)$ is completely positive if and only if $a\leq \frac{1}{\eta}(1-\sqrt{1-\eta})=:g(\eta)$
(This can be seen by applying Eq.\ (9) of Ref.\ \cite{Ru02} to our case).
Denoting by $\ell_{\mathrm{sol}}$ the solution to the equation
$a^{-\ell} = g(\eta^{-\ell})$, we see that  $E(a,\eta)$ is $\lfloor \ell_{\mathrm{sol}} \rfloor$-divisible, but   not $(\lfloor \ell_{\mathrm{sol}} \rfloor+1)$-divisible.
Correspondingly, the state can only be $n$-refined $ \lfloor\log_{p}\lfloor \ell_{\mathrm{sol}} \rfloor \rfloor$ times.
For example, for $a=0.1$ and $\eta=0.9$,  we have that the state can be $2$-refined only  $5$ times.
\end{proof}

\section{Conclusions and Outlook}
\label{sec:concl}

In summary, we have investigated which TI MPS have a CL, which is defined as the infinite iteration of the inverse of a renormalization procedure, 
together with a regularity condition in the limit.
We have found that a TI MPS has a CL if and only if its transfer matrix is infinitely divisible.
We have then defined the coarse CL 
as the CL of some of the coarser descriptions of the state,
and have characterised the states with a coarse CL using the divisibility properties of their transfer matrices.
We have shown that various well-studied states
(such as the AKLT state,
 the 1D cluster state
 or the antiferromagnet)
 have a coarse CL, but that not all states have one.

This work raises several questions.
One concerns the representation of the states obtained in the limit as matrix products, which would require a generalization of the class of cMPS.
This would also allow to study the uniqueness of the CL. 
It also remains to be seen whether there is a meaningful definition of CL such that all TI MPS have a limit of this sort.
A further  possibility is to consider the renormalization procedure  determined by the Multiscale Entanglement Renormalization Ansatz (MERA) \cite{Vi07},  for which the class of continuous MERA was defined in \cite{Ha13b}, and study continuum limits in that setting.

\section*{Acknowledgements}

GDLC thanks T. J. Osborne for discussions.
GDLC acknowledges support from the Elise Richter Fellowship of the FWF.
This work was supported in part by the Perimeter Institute of Theoretical Physics. Research at Perimeter Institute is supported by the Government of Canada through Industry Canada and by the Province of Ontario through the Ministry of Economic Development and Innovation.
N.S. acknowledges support by the European Union through the ERC-StG
WASCOSYS (Grant No. 636201).
DPG acknowledges support from  MINECO (grant MTM2014-54240-P), Comunidad de Madrid (grant QUITEMAD+-CM, ref. S2013/ICE-2801), and
Severo Ochoa project SEV-2015-556.
This work was made possible through the support of grant $\#$48322 from the John Templeton Foundation.
This project has received funding from the European Research Council (ERC) under the European Union's Horizon 2020 research and innovation programme (grant agreement No 648913).
JIC acknowledges support from the DFG through the NIM (Nano Initiative Munich).


\appendix

\section{Proof of \cref{thm:CL}}
\label{app:proof-dyadic}

Here we prove \cref{thm:CL}, which we state again.

\begin{theorem:CL}
Given $\mathcal{V}(B)$ with $B$ in irreducible form, the following statements are equivalent:
\begin{enumerate}
\item $\mathcal{V}(B)$ has a CL.
\label{im:1}
\item $E_B$ is infinitely divisible.
\label{im:2}
\item There is a TPCPM $P$ and a Liouvillian of Lindblad form $L$ such that $E_B=Pe^L$,  $P^2=P$ and $PLP=PL$. \label{im:3}
\end{enumerate}
\end{theorem:CL}

\begin{proof}
That  \cref{im:2} and  \cref{im:3} are equivalent was proven by Holevo \cite{Ho87} and Denisov \cite{De89b}.

By \cref{def:CL} and the subsequent discussion, $\mathcal{V}(B)$ has a CL if there is $p>1$ such that $E_B$ is $p$-infinitely divisible and $f_{p,E_B}$ is continuous at zero.
It is thus immediate to see that \cref{im:2} implies  \cref{im:1},
since being $p$-infinitely divisible is a particular case of being  infinitely divisible, and using \cref{im:3} we have that
$f_{p,E_B}(n/p^\ell) = Pe^{L n/p^\ell}$  is continuous at 0.

Finally, to see that \cref{im:1} implies \cref{im:2}, assume that $E_B$ is
$p$-infinitely divisible and that $f_{p,E_B}$ is continuous at 0.  We will
construct the $n$th root of $E\equiv E_B$ by using the expansion of  $1/n$ in terms
of $1/p^\ell$. So for an arbitrary $n\in \mathbb{N}$,
we have that
\be
\frac{1}{n} = \frac{1}{p^{\ell}} \left(\left\lfloor  \frac{p^{\ell}}{n}
\right\rfloor + \frac{r_\ell}{n}\right) ,\label{eq:n}
\ee
where $\lfloor  \frac{p^{\ell}}{n} \rfloor$ is the largest integer which is at most that number (floor), and $0\leq r_k<n$ is the residue of the division.

Let us  consider
\be
\left(E_{p^{\ell}}^{\lfloor p^{\ell}/n\rfloor} \right)_\ell .
\ee
Since this is a sequence in a compact space, there must exist a
subsequence that converges to a limit which we call $E_n$,
\be
\left(E_{p^{l_{k}}}^{\lfloor p^{l_{k}}/n\rfloor}=:T_k \right)_k \to E_n .
\label{eq:conv}
\ee
By completeness, $E_n$ is a quantum channel. In the rest of the proof we will show that $E_n$ is an $n$th root of $E$, i.e.\ $E_n^n = E$.

To see this, observe that
 \be
 \label{for5}
 \left| \left|E_n^n - E\right|\right| \le  \left| \left| E_n^n-T_{k}^n\right|\right| + \left|\left|T_{k}^n-E\right|\right| ,
 \ee
where for a superoperator $L$ we use the norm $||L|| = \sup_X
||L(X)||_1/||X||_1$, where $||X||_1$ denotes the Schatten 1-norm.  The
first term of \eqref{for5} vanishes as $k\to \infty$, since
 \be
 ||E_n^n - T_{k}^n || \le
n ||E_n-T_k|| \: \le
n \epsilon  ,
\label{eq:boundn}
 \ee
 where the first inequality follows from the identity $T_{k}^n - E_n^n =
(T_{k}^{n-1} + T_{k}^{n-2}E_{n} +\ldots+ E_{n}^{n-1})(T_{k}-E_{n})$ and
the fact that $||T_{k}^{n-j} E_{n}^{j-1}||=1$ for all $j=1,\ldots, n$,
 and the second from \eqref{eq:conv}.

To show that the second term of \eqref{for5} vanishes, 
we use that
\begin{align*}
\left\|T_{k}^n-E\right\|
&\stackrel{\eqref{eq:n}}{\le}
\left\|E_{p^{\ell_k}}^{\lfloor p^{\ell_k}/n\rfloor n } -
E_{p^{\ell_k}}^{\lfloor p^{\ell_k}/n\rfloor n +r_k}\right\|
\\
&\le \left\| E_{p^{\ell_k}}^{\lfloor p^{\ell_k}/n \rfloor n-1} \right\|
 \,\left\| E_{p^{\ell_k}} - E_{p^{\ell_k}}^{r_k+1}\right\|
 \leq \left\| E_{p^{\ell_k}} - E_{p^{\ell_k}}^{r_k+1}\right\| \ , 
\end{align*}
where we have used that  $\left\| E_{p^{\ell_k}}\right\|=1$. Since
$r_k+1\leq n$, we have that both $1/p^{\ell_k}$ and 
$(r_k+1)/p^{\ell_k}\rightarrow 0$, and thus continuity of 
$f_{p,E_B}\left(n,p^{\ell_k}\right) = E^{n}_{p^{\ell_k}}$
at zero implies that
$\left\| E_{p^{\ell_k}} - E_{p^{\ell_k}}^{r_k+1}\right\|\rightarrow 0$.
\end{proof}


\end{document}